\def\draft{0}

\documentclass[11pt]{article}
\usepackage{amsfonts, amsmath, amssymb}
\usepackage{url, graphics}
\usepackage{fullpage}
\usepackage[dvips]{graphicx}

\newtheorem{theorem}{Theorem}

\newtheorem{definition}{Definition}

\newtheorem{proposition}{Proposition}

\newtheorem{remk}{Remark}

\def\FullBox{\hbox{\vrule width 8pt height 8pt depth 0pt}}

\def\qed{\ifmmode\qquad\FullBox\else{\unskip\nobreak\hfil
\penalty50\hskip1em\null\nobreak\hfil\FullBox
\parfillskip=0pt\finalhyphendemerits=0\endgraf}\fi}

\newenvironment{proof}{\begin{trivlist} \item {\bf Proof:~~}}
  {\qed\end{trivlist}}

\def\qedsketch{\ifmmode\Box\else{\unskip\nobreak\hfil
\penalty50\hskip1em\null\nobreak\hfil$\Box$
\parfillskip=0pt\finalhyphendemerits=0\endgraf}\fi}

\newenvironment{proofof}[1]{\begin{trivlist} \item {\bf Proof
#1:~~}}
  {\qed\end{trivlist}}

\newcommand{\ie} {{\it i.e.\ }}
\newcommand{\eg} {{\it e.g.\ }}

\newcommand{\N}{{\mathbb{N}}}

\newcommand{\zo}{\{0,1\}}

\newcommand{\pr}[1]{\Pr\left[#1\right]}

\newcommand{\eps}{\varepsilon}

\ifnum\draft=1
\newcommand{\authnote}[2]{{ \bf [#1's Note: #2]}}
\else
\newcommand{\authnote}[2]{}
\fi

\newcommand{\COMMENT}[1]{}
\newcommand{\ket}[1]{|#1\rangle}

\newcommand{\ketbra}[2]{|#1\rangle\langle#2|}
\def\01{\{0,1\}}
\newcommand{\braket}[2]{\langle{#1}|{#2}\rangle} 
\def\01{\{0,1\}}

\newcommand{\triple}[3]{\langle{#1}|{#2}|{#3}\rangle}

\newcommand{\Tr}{\mbox{\rm Tr}}

\newcommand{\spa}[1]{\mathcal{#1}}

\newcommand{\Abort}{{\mathrm{Abort}}}

\title{Optimal bounds for quantum bit commitment}
\author{Andr\'e Chailloux$^*$ \\
LRI\\
Universit\'e Paris-Sud\\
andre.chailloux@lri.fr\\
\and
Iordanis Kerenidis\thanks{Supported in part by ANR CRAQ and AlgoQP grants of the French Ministry and in part by the European Commission under the Integrated Project Qubit Applications (QAP) funded by the IST directorate as Contract Number 015848.}\\
CNRS - LIAFA\\
Universit\'e Paris 7\\
jkeren@liafa.jussieu.fr}

\begin{document}

\maketitle


\begin{abstract}

Bit commitment is a fundamental cryptographic primitive with numerous applications. Quantum information allows for bit commitment schemes in the information theoretic setting where no dishonest party can perfectly cheat.
The previously best-known quantum protocol by Ambainis achieved a cheating probability of
at most $3/4$~\cite{Amb01}. On the other hand, Kitaev showed that no quantum protocol can have cheating probability less than $1/\sqrt{2}$~\cite{Kit03} (his lower bound on coin flipping can be easily extended to bit commitment). Closing this gap has since been an important and open question.

In this paper, we provide the optimal bound for quantum bit commitment. We first show a lower bound of approximately $0.739$, improving Kitaev's lower bound. We then present an optimal quantum bit commitment protocol which has cheating probability arbitrarily close to $0.739$.  More precisely, we show how to use any weak coin flipping protocol with cheating probability $1/2 + \eps$ in order to achieve a quantum bit commitment protocol with cheating probability $0.739 + O(\eps)$. We then use the optimal quantum weak coin flipping
protocol described by Mochon~\cite{Moc07}. To stress the fact that our protocol uses quantum effects beyond the weak coin flip, we show that any classical bit commitment protocol with access to perfect weak (or strong) coin flipping has cheating probability at least $3/4$.
\end{abstract}

\section{Introduction}

Quantum information has given us the opportunity to revisit information theoretic security in cryptography. The first breakthrough result was a protocol of Bennett and Brassard~\cite{BB84} that showed how to securely distribute a secret key between two players in the presence of an omnipotent eavesdropper. Thenceforth, a long series of work has focused on which other cryptographic primitives are possible with the help of quantum information. Unfortunately, the subsequent results were not positive. Mayers and Lo, Chau proved the impossibility of secure quantum bit commitment and oblivious transfer and consequently of any type of two-party secure computation~\cite{May97,LC97,DKSW07}. However, several weaker variants of these primitives have been shown to be possible~\cite{HK04,BCH+08}.

The main primitives that have been studied are coin flipping, bit commitment and oblivious transfer. Coin flipping is a cryptographic primitive that enables two distrustful and far apart parties, Alice and Bob, to create a random bit that remains unbiased even if one of the players tries to force a specific outcome. It was first proposed by Blum ~\cite{Blu81} and has since found numerous applications in two-party secure computation.
In the classical world, coin flipping is possible under computational assumptions like the hardness of factoring or the discrete log problem. However, 
in the information theoretic setting, it is not hard to see that in any classical protocol, one of the players can always bias the coin to his or her desired outcome with probability 1.  

Aharonov et al.~\cite{ATVY00} provided a quantum protocol where no dishonest player could bias the coin with probability higher than 0.9143. Then, Ambainis~\cite{Amb01} described an improved protocol whose cheating probability was at most $3/4$. Subsequently, a number of different protocols have been proposed~\cite{SR01,NS03,KN04} that achieved the same bound of $3/4$. On the other hand, Kitaev~\cite{Kit03}, using a formulation of quantum coin flipping protocols as semi-definite programs proved a lower bound of $1/2$ on the product of the two cheating probabilities for Alice and Bob (for a proof see \eg ~\cite{ABD+04}). In other words, no quantum coin flipping protocol can achieve a cheating probability less than $1/\sqrt{2}$ for both Alice and Bob. Recently, we resolved the question of whether $3/4$ or $1/\sqrt{2}$ is ultimately the right bound for quantum coin flipping by constructing a strong coin-flipping protocol with cheating probability $1/\sqrt{2}+\eps$ (\cite{CK09}).  

The protocol in \cite{CK09} is in fact a {\em classical} protocol that uses the primitive of  weak coin flipping as a subroutine. In the setting of weak coin flipping, Alice and Bob have a priori a desired coin outcome, in other words the two values of the coin can be thought of as `Alice wins' and `Bob wins'. We are again interested in bounding the probability that a dishonest player can win this game. Weak coin flipping protocols with cheating probabilities less than $3/4$ were constructed in ~\cite{SR02,Amb02,KN04}. Finally, a breakthrough result by Mochon resolved the question of the optimal quantum weak coin flipping. First, he described a protocol with cheating probability $2/3$~\cite{Moc04,Moc05} and then a protocol that achieves a cheating probability of $1/2+\eps$ for any $\eps>0$ ~\cite{Moc07}.
 
In other words, in coin flipping, the power of quantum really comes from the ability to perform weak coin flipping. If there existed a classical weak coin flipping protocol with arbitrarily small bias, then this would have implied a classical strong coin flipping protocol with cheating probability arbitrarily close to $1/\sqrt{2}$ as well.

In this paper, we turn our attention to bit commitment. Even though this primitive is closely related to coin flipping we will see that actually the results are surprisingly different. A bit commitment protocol consists of two phases: in the commit phase, Alice commits to a bit $b$; in the reveal phase, Alice reveals the bit to Bob. We are interested in the following two probabilities: Alice's cheating probability is the average probability of revealing both bits during the reveal phase, and Bob's cheating probability is the probability he can guess the bit $b$ after the commit phase. 

Using the known results about coin flipping we can give the following bounds on these probabilities. First, most of the suggested coin flipping protocols with cheating probability $3/4$ were using some form of imperfect bit commitment scheme. More precisely, Alice would quantumly commit to a bit $a$, Bob would announce a bit $b$ and then Alice would reveal her bit $a$. The outcome of the coin flip would be $a \oplus b$. Hence, we already know bit commitment protocols that achieve cheating probability equal to $3/4$. Note also that Ambainis had proved a lower bound of $3/4$ for any protocol of this type. On the other hand, a bit commitment protocol with cheating probability $p$ immediately gives a strong coin flipping protocol with the same cheating probability (by the above mentioned construction) and hence Kitaev's lower bound of $1/\sqrt{2}$ still holds. 

The question of the optimal cheating probability for bit commitment remained unresolved, similar to the case of coin flipping that was answered in \cite{CK09}. Here, we find the optimal cheating probability for quantum bit commitment, which surprisingly is neither of the above mentioned constants. In fact, we show that it is approximately $0.739$. 

We start by providing a lower bound for any quantum bit commitment protocol. In order to do so, we describe an explicit cheating strategy for Alice and Bob in any protocol. In high level, let $\ket{\psi_b}$ be the joint state of Alice and Bob after the commit phase and $\sigma_b$ Bob's density matrix, when Alice honestly commits to bit $b$. It is well known that there exists a cheating strategy for Bob that succeeds with probability
\[ P^*_B \geq \frac{1}{2}+\frac{\Delta(\sigma_0,\sigma_1)}{2}
\]
where $\Delta(\cdot , \cdot)$ denotes the trace distance between two density matrices.

For Alice, we consider the following cheating strategy. Instead of choosing a bit $b$ in the beginning of the protocol, she goes into a uniform superposition of the two possible values and controlled on this qubit she performs honestly the commit phase. Then, after the commit phase, when she wants to reveal a specific bit $b$, she first performs a unitary operation on her part to transform the joint state to one which is as close as possible to the honest state $\ket{\psi_b}$ (the unitary is given by Uhlmann's theorem) and then performs the reveal phase honestly. 

It is not hard to see that Alice's cheating probability is at least
\[
P^*_A \ge \frac{1}{2}\left( F^2(\sigma_+,\sigma_{0}) +  F^2(\sigma_+,\sigma_{1})\right)
\]
where $F(\cdot , \cdot)$ denotes the fidelity between two states and $\sigma_+ = \frac{1}{2}\left(\sigma_0 + \sigma_1\right)$.

 In order to conclude we prove our main technical lemma
\begin{proposition} 
Let $\sigma_0,\sigma_1$ any two quantum states. Let $\sigma_+ = \frac{1}{2}\left(\sigma_0 + \sigma_1\right)$. We have
\begin{align*}
 \frac{1}{2}\left( F^2(\sigma_+,\sigma_{0}) +  F^2(\sigma_+,\sigma_{1})\right) \ge (1 - (1 - \frac{1}{\sqrt{2}})\Delta(\sigma_0,\sigma_1))^2
\end{align*}
\end{proposition}

By equalizing the two lower bounds that are expressed in terms of the trace distance we conclude that 
\begin{theorem}\label{LowerBound}
In any quantum bit commitment protocol with cheating probabilities $P_A^*$ and $P_B^*$ we have 
$\max\{P^*_A,P^*_B\} \ge 0.739$.
\end{theorem}

Then, we provide a matching upper bound. We describe a quantum bit commitment protocol that achieves a cheating probability arbitrarily close to $0.739$. Out protocol uses a weak coin flipping protocol with cheating probability $1/2 + \epsilon$ as a subroutine and achieves a cheating probability for the bit commitment of $0.739 + O(\epsilon)$. 

\begin{theorem}
There exists a quantum bit commitment protocol that uses a weak coin flipping protocol with cheating probability $1/2+\epsilon$ as a subroutine and achieves cheating probabilities less than $0.739+O(\epsilon)$.
\end{theorem}
 
We note that our protocol is in fact quantum even beyond the weak coin flip subroutine. This is in fact necessary. We show that any classical bit commitment protocol with access to a perfect weak coin (or even strong) cannot achieve cheating probability less than $3/4$. 

\begin{theorem}\label{LowerBoundClassical}
Any classical bit commitment protocol with access to perfect weak (or strong) coin flipping cannot achieve cheating probabilities less than $3/4$.
\end{theorem}

Unlike the case of quantum strong coin flipping that is derived classically when one has access to a weak coin flipping protocol, the optimal quantum bit commitment takes advantage of quantum effects beyond the weak coin flipping subroutine.


\setcounter{theorem}{0}

\section{Preliminaries}

\subsection{Useful facts about trace distance and fidelity of quantum states}

We start by stating a few properties of the trace distance $\Delta$ and fidelity $F$ between two quantum states.

\begin{definition}
For any two quantum states $\rho,\sigma$, the trace distance $\Delta$ between them is given by $\Delta(\rho,\sigma) = \Delta(\sigma,\rho) = \frac{1}{2}tr(|\rho - \sigma|)$ where $|A| = \sqrt{A^{\dagger}A}$ for a matrix $A$
\end{definition}

\begin{proposition}~\label{SumTraceDistance}
For any two states $\rho,\sigma$ such that $\rho = \sum_{i} p_i \ketbra{i}{i}$ and $\sigma = \sum_i q_i \ketbra{i}{i}$, we have
\begin{align*}
\Delta(\rho,\sigma) = \sum_i \frac{1}{2} |p_i - q_i| = \sum_{i: p_i \ge q_i} (p_i - q_i) = 1 - \sum_i \min\{p_i,q_i\} = \sum_i \max\{p_i,q_i\} - 1
\end{align*}
\end{proposition}
\begin{proof}
Since $\sum_i p_i = \sum_i q_i = 1$, we have
$
\sum_{i: p_i \ge q_i} (p_i - q_i) = \sum_{i p_i < q_i} (q_i - p_i)$ and $\sum_i \max\{p_i,q_i\} + \min\{p_i,q_i\} = 2$ hence
\begin{align*}
\Delta(\rho,\sigma) = \sum_i \frac{1}{2} |p_i - q_i| & = \frac{1}{2} \left(\sum_{i: p_i \ge q_i} (p_i - q_i) + \sum_{i: p_i < q_i} (q_i - p_i)\right) =  \sum_{i: p_i \ge q_i} (p_i - q_i) \\
\Delta(\rho,\sigma) = \sum_i \frac{1}{2} |p_i - q_i| & = \frac{1}{2} \sum_i (\max\{p_i,q_i\} - \min\{p_i,q_i\}) = 1 - \sum_i \min\{p_i,q_i\} = \sum_i \max\{p_i,q_i\} - 1
\end{align*}
\end{proof}

\begin{proposition}~\label{PropPOVM}
For any two states $\rho,\sigma$, and a POVM $E = \{E_1,\dots,E_m\}$ with $p_i = tr(\rho E_i)$ and $q_i = tr(\sigma E_i)$, we have $\Delta(\rho,\sigma) \ge \frac{1}{2} \sum_{i} |p_i - q_i|$. There is a POVM (even a projective measurement) for which this inequality is an equality.
\end{proposition}

\begin{proposition}~\cite{Hel67}\label{Hel67}
Suppose Alice has a bit $c \in_R \zo$ unknown to Bob. Alice sends a quantum state $\rho_c$ to Bob. We have
\[
\Pr[\mbox{Bob guesses } c] \le \frac{1}{2} + \frac{\Delta(\rho_0,\rho_1)}{2}
\]
\end{proposition}

\begin{definition}
For any two states $\rho,\sigma$, the fidelity $F$ between them is given by
$F(\rho,\sigma) = F(\sigma,\rho)= tr(\sqrt{\rho^{\frac{1}{2}}\sigma\rho^{\frac{1}{2}}}) $
\end{definition}

\begin{proposition}\label{POVMfidelity}
For any two states $\rho,\sigma$, and a POVM $E = \{E_1,\dots,E_m\}$ with $p_i = tr(\rho E_i)$ and $q_i = tr(\sigma E_i)$, we have $F(\rho,\sigma) \le \sum_{i} \sqrt{p_i q_i}$. There is a POVM for which this inequality is an equality.
\end{proposition}

\begin{proposition}[Uhlmann's theorem]
For any two states $\rho,\sigma$, there exist a purification $\ket{\phi}$ of $\rho$ and a purification $\ket{\psi}$ of $\sigma$ such that $|\braket{\phi}{\psi}| = F(\rho,\sigma)$
\end{proposition}

\begin{proposition}\label{CPTPfidelity}
For any two states $\rho,\sigma$ and a completely positive trace preserving operation $Q$, we have $F(\rho,\sigma) \le F(Q(\rho),Q(\sigma))$.
\end{proposition}

\subsection{Definition of quantum bit commitment}

\begin{definition}
A quantum commitment scheme  is an interactive protocol between Alice and Bob with two phases, a Commit phase and a Reveal phase.
\begin{itemize}
\item In the {\em commit} phase, Alice interacts with Bob   in order to commit to $b$.
\item In the {\em reveal} phase, Alice interacts with Bob in order to reveal $b$. Bob decides to accept or reject depending on the revealed value of $b$ and his final state. We say that Alice successfully reveals $b$, if Bob accepts the revealed value.
\end{itemize}

\noindent We define the following security requirements for the commitment scheme.
\begin{itemize}
\item \emph{Completeness:} If Alice and Bob are both honest then Alice always successfully reveals the bit $b$ she committed to.
\item  \emph{Binding property:} For any cheating Alice and for honest Bob, we define Alice's cheating probability as
\begin{align*}
P^*_A = \frac{1}{2}\left(
\Pr[\mbox{ Alice successfully reveals } b = 0 ] + 
\Pr[\mbox{ Alice successfully reveals } b = 1 ]\right)
\end{align*}
\item \emph{Hiding property:} For any cheating Bob and for honest Alice, we define Bob's cheating probability as
\begin{align*}
P^*_B = 
\Pr[\mbox{ Bob guesses } b \mbox{ after the Commit phase }] 
\end{align*}
\end{itemize}
\end{definition}

\paragraph{Remark:} 

The definition of quantum bit commitment we use is the standard one when one studies stand-alone cryptographic primitives. In this setting, quantum bit commitment has a clear relation to other fundamental primitives such as coin flipping and oblivious transfer~\cite{ATVY00,Amb01,Kit03,Moc07,CKS10}. Moreover, the study of such primitives sheds light on the physical limits of quantum mechanics and the power of entanglement. Recently there have been some stronger definitions of Quantum Bit Commitment protocols that suit better practical uses (see for example~\cite{DFR+07}). 
\COMMENT{
However, since we are interested in the stand alone primitive and the physical bounds concerning the feasibility of this primitive, we keep the original quantum bit commitment definition. This preserves the spirit of bit commitment in the sense that a cheating player should not be able to choose which bit he reveals {\it ie } that he cannot decommit to both bits without being caught.}
Notice that using our weaker definition of quantum bit commitment only strengthens our lower bound which also holds for the stronger ones.


We now describe more in detail the different steps on a quantum bit commitment protocol. We consider protocols where Alice reveals $b$ at the beginning of the decommit phase. Note that this doesn't help Bob and can only harm a cheating Alice. Proving a lower bound for such protocols will hence be a lower bound for all bit commitment protocols. 

We assume here that Alice and Bob are both honest. Let $\spa{A}$ Alice's space and $\spa{B}$ Bob's space.

\paragraph{The commit phase:} Alice wants to commit to a bit $b$. Alice and Bob communicate with each other and perform some quantum operations. This can be seen as a joint quantum operation which depends on $b$. We can suppose wlog that this operation is a quantum unitary $U^C_b$ (by increasing Alice and Bob's quantum space). At the end of the commit phase, Alice and Bob share the quantum state $\ket{\psi_b}$. Let $\sigma_b = \Tr_{\spa{A}} \ketbra{\psi_b}{\psi_b}$ the state that Bob has after the commit phase.

\paragraph{The reveal phase:} Alice wants to reveal $b$ to Bob. Alice reveals $b$ at the beginning of the decommit phase.
 Similarly to the commit phase, we can suppose that the decommit phase is equivalent to Alice and Bob performing a joint unitary $U^D_b$ on their shared state ($\ket{\psi_b}$ if they were honest in the Commit phase). At the end, Bob performs a check to see whether Alice cheated or not. In the honest case, Bob always accepts.

\subsection{Definitions of Coin flipping} \label{coinDefs}

We provide the formal definitions of all the different variants of coin flipping protocols that we are going to use.

In a coin flipping protocol, we call a round of communication one message from Alice to Bob and one message from Bob to Alice. We suppose that Alice always sends the first message and Bob always sends the last message. The protocol is quantum if we allow the parties to send quantum messages and perform quantum operations. A player is honest if he or she follows the protocol. A cheating player can deviate arbitrarily from the protocol but still outputs a value at the end of it.
There are two important variants of coin flipping that have been studied.

\paragraph{Strong Coin Flipping\\}
A strong coin flipping protocol between two parties Alice and Bob is a protocol where Alice and Bob interact and at the end, Alice outputs a value $c_A \in \{0,1,\Abort\}$ and Bob outputs a value $c_B \in \{0,1,\Abort\}$. If $c_A = c_B$, we say that the protocol outputs $c = c_A$. If $c_A \neq c_B$ then the protocol outputs $c = \Abort$.

A strong coin flipping protocol with bias $\varepsilon$ ($SCF(\varepsilon)$) has the following properties
\begin{itemize}
\item If Alice and Bob are honest then $\pr{c = 0} = \pr{c=1} = 1/2$
\item If Alice cheats and Bob is honest then $P^*_A = \max\{\pr{c = 0},\pr{c = 1}\} \le 1/2 + \varepsilon$.
\item If Bob cheats and Alice is honest then $P^*_B = \max\{\pr{c = 0},\pr{c = 1}\} \le 1/2 + \varepsilon$
\end{itemize} 
The probabilities $P^*_A$ and $P^*_B$ are called the cheating probabilities of Alice and Bob respectively. The cheating probability of the protocol is defined as $\max\{P^*_A,P^*_B\}$. We say that the coin flipping is \emph{perfect} if $\eps = 0$. This is because a player that want to Abort can always declare victory rather than aborting without reducing the security of the protocol(see~\cite{Moc07}).

\paragraph{Weak coin flipping\\}

A weak coin flipping protocol between two parties Alice and Bob is a protocol where Alice and Bob interact and at the end, Alice outputs a value $c_A \in \{0,1\}$ and Bob outputs a value $c_B \in \{0,1\}$. If $c_A = c_B$, we say that the protocol outputs $c = c_A$. If $c_A \neq c_B$ then the protocol outputs $c = \Abort$. The difference with Strong coin flipping is that the players do not Abort. This is because a player that wants to Abort can always declare victory rather than aborting without reducing the security of the protocol.

A (balanced) weak coin flipping protocol with bias $\varepsilon$ ($WCF(1/2,\varepsilon)$) has the following properties
\begin{itemize}
\item If $c = 0$, we say that Alice wins. If $c = 1$, we say that Bob wins.
\item If Alice and Bob are honest then $\pr{\mbox{ Alice wins }} = \pr{ \mbox{ Bob wins }} = 1/2$
\item If Alice cheats and Bob is honest then $P^*_A = \pr{ \mbox{ Alice wins }} \le 1/2 + \varepsilon$
\item If Bob cheats and Alice is honest then $P^*_B = \pr{ \mbox{ Bob wins }}\le 1/2 + \varepsilon$
\end{itemize}
Similarly, $P^*_A$ and $P^*_B$ are the cheating probabilities of Alice and Bob. The cheating probability of the protocol is defined as $\max\{P^*_A,P^*_B\}$.

 We can also define weak coin flipping for the case where the winning probabilities of the two players in the honest case are not equal.

\paragraph{Unbalanced weak coin flipping\\} 
A weak coin flipping protocol with parameter $z$ and bias $\varepsilon$ ($WCF(z,\varepsilon))$ has the following properties.

\begin{itemize}
\item If $c = 0$, we say that Alice wins. If $c = 1$, we say that Bob wins.
\item If Alice and Bob are honest then $\pr{\mbox{ Alice wins }} = z$ and $\pr{\mbox{ Bob wins }} = 1 - z$
\item If Alice cheats and Bob is honest then $P^*_A = \pr{\mbox{ Alice wins }} \le z + \varepsilon$
\item If Bob cheats and Alice is honest then  $P^*_B = \pr{\mbox{ Bob wins }} \le (1 - z) + \varepsilon$
\end{itemize}

\paragraph{Reformulation of Quantum weak coin flipping protocol}
We reformulate here the definition of a quantum weak coin flipping to take into account the fact that Alice and Bob are quantum players that perform unitary operations during the protocol and at the end they perform a measurement on a quantum register in order to get their classical output. More precisely, let $\spa{O}_{A}$ (resp. $\spa{O}_B$) be Alice's (resp. Bob's) one-qubit output register. At the end of the protocol Alice (resp. Bob) has a state $\rho_A$ in $\spa{O}_A$ ( resp. $\rho_B$ in $\spa{O}_B$ ). They also share some garbage state. The players get their output value by measuring their output qubit in the computational basis. Let $\rho_{AB}$ the joint output state of Alice and Bob in $\spa{O}_A \otimes \spa{O}_B$. In this setting, a weak coin flipping has the following properties.
\begin{itemize}
\item The $0$ outcome corresponds to Alice winning. The $1$ outcome corresponds to Bob winning.
\item If Alice and Bob are honest then $\triple{00}{\rho_{AB}}{00} = \triple{11}{\rho_{AB}}{11} = 1/2$
\item If Alice cheats and Bob is honest then $P^*_A = \triple{0}{\rho_B}{0} \le 1/2 + \varepsilon$
\item If Bob cheats and Alice is honest then $P^*_B = \triple{1}{\rho_A}{1} \le 1/2 + \varepsilon$
\end{itemize}

Notice that Alice's cheating probability depends only on Bob's output. This is because a cheating Alice will always claim that she won, so she wins when Bob outputs `Alice wins'. We have the same behavior for a cheating Bob.

Similarly, we can define an unbalanced weak coin flipping in this setting.
\begin{itemize}
\item The $0$ outcome corresponds to Alice winning. The $1$ outcome corresponds to Bob winning.
\item If Alice and Bob are honest then $\triple{00}{\rho_{AB}}{00} = z$ ; $\triple{11}{\rho_{AB}}{11} = 1-z$
\item If Alice cheats and Bob is honest then $P^*_A = \triple{0}{\rho_B}{0} \le z + \varepsilon$
\item If Bob cheats and Alice is honest then $P^*_B = \triple{1}{\rho_A}{1} \le (1-z) + \varepsilon$
\end{itemize}

We will use the following result by Mochon.
\begin{proposition}{\em\cite{Moc07}}\label{Mochon}
For every $\varepsilon > 0$, there exists a quantum $WCF(1/2,\varepsilon)$ protocol $P$. 
\end{proposition}

Note also that this construction can be extended to the unbalanced case. A procedure to use balanced $WCF$ protocols to unbalanced ones has been presented in~\cite{CK09}. This procedure was presented in the classical setting but can be easily extended to the quantum definitions of unbalanced weak coin.
\begin{proposition}[CK09]\label{CK09}
Let  $P$ be a $WCF(1/2,\varepsilon)$ protocol with $N$ rounds. Then, 
$\forall z \in [0,1]$ and $\ \forall k \in \N$, there exists a  $WCF(x,\varepsilon_0)$ protocol $Q$ such that:
\begin{itemize}
\item $Q$ uses $k\cdot N$ rounds.
\item $|x - z| \le 2^{-k}$.
\item $\varepsilon_0 \le 2\varepsilon$.
\end{itemize}
\end{proposition}


\section{Proof of the Lower Bound}
To prove the lower bound, we will show some generic cheating strategies for Alice and Bob that work for any kind of bit commitment scheme. We will then show that these cheating strategies give  a cheating probability of approximately 0.739 for any protocol. 
\subsection{Description of cheating strategies}
We denote by $\ket{\psi_b}$  the quantum state Alice and Bob share at the end of the commit phase. Let $\sigma_b = \Tr_{\spa{A}} \ketbra{\psi_b}{\psi_b}$ the state that Bob has after the commit phase when Alice honestly commits to bit $b$.
\subsubsection{Bob's cheating strategy}
The cheating strategy of Bob is the following:
\begin{itemize}
\item Perform the Commit phase honestly.
\item Guess $b$ by performing on the state at the end of the commit phase the optimal discriminating measurement between $\sigma_0$ and $\sigma_1$.  
\end{itemize}
First note that an all-powerful Bob can always perform this strategy, since he knows the honest states $\sigma_0$ and $\sigma_1$ and can hence compute and perform the optimal measurement. Let us analyze this strategy. We know~\cite{Hel67} that Bob can guess $b$ with probability $\frac{1}{2} + \frac{\Delta(\sigma_0,\sigma_1)}{2}$ and hence
\[
P^*_B \ge \frac{1}{2} + \frac{\Delta(\sigma_0,\sigma_1)}{2}
\]
\subsubsection{Alice's cheating strategy}
The cheating strategy of Alice is the following
\begin{itemize}
\item Perform a quantum strategy so that at the end of the commit phase, Bob has the state $\sigma_+ = \frac{1}{2}\left(\sigma_0 + \sigma_1\right)$. 
\item In order to reveal a specific value $b$, send $b$ then apply a local quantum operation such that the actual joint state of the protocol, $\ket{\phi_{b}}$, satisfies $|\braket{\phi_{b}}{\psi_{b}}| = F(\sigma_+,\sigma_{b})$. Perform the rest of the reveal phase honestly.
\end{itemize}

First note that an all-powerful Alice can perform this strategy. An honest Alice has a strategy to make Bob's state after the commit phase equal to $\sigma_b$ for both $b=0$ and $b=1$. A cheating Alice creates a qubit $\frac{1}{\sqrt{2}}(\ket{0}+\ket{1})$. Conditioned on 0 (resp. 1), she applies the strategy that will give Bob the state $\sigma_0$ (resp. $\sigma_1$). By doing this Bob's state at the end of the commit phase is exactly $\sigma_{+}$.
Moreover, by Uhlmann's theorem, Alice can compute and perform the local unitary in the beginning of the reveal phase to create a state $\ket{\phi_{b}}$ that satisfies $|\braket{\phi_{b}}{\psi_{b}}| = F(\sigma_+,\sigma_{b})$. 

For the analysis, since Bob accepts $b$ with probability $1$ when the joint state of the protocol is $\ket{\psi_b}$, he accepts with probability at least $|\braket{\phi_{b}}{\psi_{b}}|^2 = F^2(\sigma_+,\sigma_{b})$ when the joint state of the protocol is $\ket{\phi_b}$.
From this cheating strategy, we have that 
\[
P^*_A \ge \frac{1}{2}\left( F^2(\sigma_+,\sigma_{0}) +  F^2(\sigma_+,\sigma_{1})\right)
\]

\subsection{Showing the Lower Bound}

We have the following bounds for cheating Alice and cheating Bob. 
\begin{align*}
P^*_A & \ge \frac{1}{2}\left( F^2(\sigma_+,\sigma_{0}) +  F^2(\sigma_+,\sigma_{1})\right) \\
P^*_B & \ge \frac{1}{2} + \frac{\Delta(\sigma_0,\sigma_1)}{2}
\end{align*}
We now use the following inequality that will be proved in the next section
\begin{proposition}\label{FidelityProp}
Let $\sigma_0,\sigma_1$ any two quantum states. Let $\sigma_+ = \frac{1}{2}\left(\sigma_0 + \sigma_1\right)$. We have
\begin{align*}
 \frac{1}{2}\left( F^2(\sigma_+,\sigma_{0}) +  F^2(\sigma_+,\sigma_{1})\right) \ge \left(1 - (1 - \frac{1}{\sqrt{2}})\Delta(\sigma_0,\sigma_1)\right)^2 .
\end{align*}
\end{proposition}
Let $t=\Delta(\sigma_0,\sigma_1)$. From the above Proposition, we have the following bounds.
\begin{align*}
P^*_A & \ge \frac{1}{2}\left( F^2(\sigma_+,\sigma_{0}) +  F^2(\sigma_+,\sigma_{1})\right) \ge 
\left(1 - (1 - \frac{1}{\sqrt{2}})t\right)^2\\
P^*_B & \ge \frac{1}{2} + \frac{\Delta(\sigma_0,\sigma_1)}{2} = \frac{1 + t}{2} 
\end{align*}
We get the optimal cheating probability by equalizing these two bounds, {\it ie.}
\[  \left(1 - (1 - \frac{1}{\sqrt{2}})t\right)^2 = \frac{1 + t}{2} 
\]
Notice that the same cheating probabilities appeared in the analysis of a weak coin flipping protocol in \cite{KN04}.
Solving the equation gives $t\approx0.4785$ and hence we have
\begin{theorem}
In any quantum bit commitment protocol with cheating probabilities $P_A^*$ and $P_B^*$ we have 
$\max\{P^*_A,P^*_B\} \ge 0.739$.

\end{theorem}

\subsection{Proof of the fidelity Lemma}

In this Section, we show Proposition$~\ref{FidelityProp}$.

\begin{proofof}{of Proposition$~\ref{FidelityProp}$}
We will prove this Lemma in three steps. Let $\sigma_0,\sigma_1$ two quantum states and let $\sigma_+ = \frac{1}{2}\left(\sigma_0 + \sigma_1\right)$.
\paragraph{Step 1}
We first consider the states $\rho_0 = \frac{1}{2} \ketbra{0}{0} \otimes \sigma_0 + 
\frac{1}{2} \ketbra{1}{1} \otimes \sigma_1$ and 
$\rho_+ = \frac{1}{2} \ketbra{0}{0} \otimes \sigma_+ + 
\frac{1}{2} \ketbra{1}{1} \otimes \sigma_+$. 
We compute the trace distance and fidelity of these states
\begin{align}
\Delta(\rho_0,\rho_+) & = \frac{1}{2} \left(\Delta(\sigma_0,\sigma_+) + \Delta(\sigma_1,\sigma_+)\right) = \frac{1}{2}\Delta(\sigma_0,\sigma_1)  \label{Drho} 
\end{align}
In order to calculate the fidelity we note first that $\rho_+^{\frac{1}{2}} = \frac{1}{\sqrt{2}} \left( \ketbra{0}{0} \otimes \sigma_+^{\frac{1}{{2}}} + \ketbra{1}{1} \otimes \sigma_+^{\frac{1}{{2}}}\right)$. From the definition of fidelity we have

\begin{align*}
F(\rho_0,\rho_+) & = tr\left(\sqrt{\rho_+^{\frac{1}{2}} \rho_0 \rho_+^{\frac{1}{2}}}\right) \\
 & = tr\left(\sqrt{\frac{1}{4}\ketbra{0}{0} \otimes \sigma_+^{\frac{1}{2}}\sigma_0 \sigma_+^{\frac{1}{2}} + \frac{1}{4}\ketbra{1}{1} \otimes \sigma_+^{\frac{1}{2}}\sigma_1 \sigma_+^{\frac{1}{2}}  }\right) \\
& = tr\left(\frac{1}{2} \ketbra{0}{0} \otimes \sqrt{\sigma_+^{\frac{1}{2}}\sigma_0 \sigma_+^{\frac{1}{2}}} + \frac{1}{2}\ketbra{1}{1} \otimes \sqrt{\sigma_+^{\frac{1}{2}}\sigma_1 \sigma_+^{\frac{1}{2}}} \right) \\
& = \frac{1}{2} tr\left(\sqrt{\sigma_+^{\frac{1}{2}}\sigma_0 \sigma_+^{\frac{1}{2}}}\right) + 
 \frac{1}{2} tr\left(\sqrt{\sigma_+^{\frac{1}{2}}\sigma_1 \sigma_+^{\frac{1}{2}}}\right) \\
& = \frac{1}{2}\left(F(\sigma_0,\sigma_+) + F(\sigma_1,\sigma_+)\right)
\end{align*} 
Hence, by Cauchy-Schwartz we conclude that
\begin{align}
F^2(\rho_0,\rho_+) & \le \frac{1}{2} F^2(\sigma_0,\sigma_+) + \frac{1}{2} F^2(\sigma_1,\sigma_+)\label{Frho}
\end{align}

\paragraph{Step 2}

Consider the POVM $E = \{E_1,\dots,E_m\}$ with $p_i = tr(\rho_0 E_i)$ and 
$q_i = tr(\rho_+ E_i)$ such that $F(\rho_0,\rho_+) = \sum_{i} \sqrt{p_i q_i}$ (Prop. \ref{POVMfidelity}). We consider the states $D_0 = \sum_{i} p_i \ketbra{i}{i}$ and $D_+ = \sum_{i} q_i \ketbra{i}{i}$. For the trace distance and 
fidelity of these states, we have
\begin{align}
\Delta(D_0,D_+) & = \frac{1}{2} \sum_i |p_i - q_i|  \le \Delta(\rho_0,\rho_+) = \frac{1}{2}\Delta(\sigma_0,\sigma_1)
& \mbox{by Prop.}~\ref{SumTraceDistance},~\ref{PropPOVM} \mbox{ and Eq.}~\ref{Drho} \\
F(D_0,D_+) & = F(\rho_0,\rho_+)  = \sum_i \sqrt{p_i q_i}
\end{align}

\paragraph{Step 3}
Let us define $k$ such that $k/2 = \Delta(D_0,D_+)$. 
We now consider the states $T_0 = k \ketbra{0}{0} + (1-k) \ketbra{2}{2}$ and 
$T_+ = \frac{k}{2} \ketbra{0}{0} + 
\frac{k}{2} \ketbra{1}{1} + (1-k) \ketbra{2}{2}$. We calculate the trace distance and fidelity of these states
\begin{align}
\Delta(T_0,T_+) &= \frac{k}{2} = \Delta(D_0,D_+) \le \frac{\Delta(\sigma_0,\sigma_1)}{2} \\
F(T_0,T_+) &= \left(1 - k + \frac{k}{\sqrt{2}}\right) \ge \left( 1- (1-\frac{1}{\sqrt{2}})\Delta(\sigma_0,\sigma_1)\right)
\end{align}
The only thing remaining is to show that $F(T_0,T_+) \le F(D_0,D_+)$. 
To prove this, we construct a completely positive trace preserving operation $Q$ such that
$Q(T_0) = D_0$ and $Q(T_+) = D_+$. We can then conclude using Proposition~\ref{CPTPfidelity}.


We define $D_1 = \sum_{i} r_i \ketbra{i}{i}$ with $p_i + r_i = 2q_i$. This means that 
$D_+ = \frac{1}{2} D_0 + \frac{1}{2} D_1$ and $\Delta(D_0,D_1) = k$.

Let $A = \{i : p_i \ge r_i\}$ and $B = \{i : p_i < r_i\}$. Let $w_i = \min\{p_i,r_i\}$ We consider the following $Q$
\begin{align*}
Q(\ketbra{0}{0}) & = \sum_{i \in A} \frac{1}{k} (p_i - r_i) \ketbra{i}{i} \\
Q(\ketbra{1}{1}) & = \sum_{i \in B} \frac{1}{k} (r_i - p_i) \ketbra{i}{i} \\
Q(\ketbra{2}{2}) & = \sum_i \frac{1}{1 - k} w_i \ketbra{i}{i} \\
Q(\ketbra{i}{j}) & = 0 \quad \mbox{ for } i \neq j
\end{align*}
Since $\Delta(D_0,D_1) = k$, we have in particular that 
$\sum_{i} w_i = 1 - k$ ;
$\sum_{i \in A} (p_i - r_i)  = \sum_{i \in B} (r_i - p_i) = k$ (see Proposition~\ref{SumTraceDistance}). $Q$ is hence a completely positive trace preserving operation. We now have:
\begin{align*}
Q(T_0) & = k \sum_{i \in A} \frac{1}{k} (p_i - r_i) \ketbra{i}{i} + (1-k) \sum_i \frac{1}{1-k} w_i \ketbra{i}{i} \\
& = \sum_{i \in A} (p_i - r_i) \ketbra{i}{i} + \sum_i w_i \ketbra{i}{i}\\
& = \sum_{i \in A} (p_i - r_i + r_i) \ketbra{i}{i} + \sum_{i \in B} p_i \ketbra{i}{i} \\
& = \sum_{i} p_i \ketbra{i}{i} = D_0
\end{align*}
Similarly, we have
\begin{align*}
Q(T_+) & = 
\frac{k}{2} \sum_{i \in A} \frac{1}{k} (p_i - r_i) \ketbra{i}{i} +\frac{k}{2}
\sum_{i \in B} \frac{1}{k}(r_i - p_i) \ketbra{i}{i} + (1-k) \sum_i \frac{1}{1-k} w_i \ketbra{i}{i} \\
& =  \sum_{i \in A} \frac{p_i - r_i}{2} \ketbra{i}{i}
+ \sum_{i \in B} \frac{r_i - p_i}{2} \ketbra{i}{i} + \sum_i w_i \ketbra{i}{i} \\
& = \sum_{i \in A} (r_i + \frac{p_i - r_i}{2}) \ketbra{i}{i} + 
\sum_{i \in B} (p_i + \frac{r_i - p_i}{2}) \ketbra{i}{i} \\
& = \sum_i q_i \ketbra{i}{i} = D_+
\end{align*}
From this, we conclude that 
\begin{align}
F(D_0,D_+) = F(Q(T_0),Q(T_+)) \ge F(T_0,T_+).
\end{align} 
Putting everything together, we have using equations (2),(4),(6),(7)
\begin{align*}
\frac{1}{2} \left(F^2(\sigma_0,\sigma_+) + F^2(\sigma_1,\sigma_+)\right)
\ge F^2(\rho_0,\rho_+) = F^2(D_0,D_+) \ge F^2(T_0,T_+) \ge \left(1 - (1 - \frac{1}{\sqrt{2}})\Delta(\sigma_0,\sigma_1)\right)^2
\end{align*} 
\end{proofof}


\section{Proof of the Upper Bound}

In this section we describe and analyze a protocol that proves the optimality of our bound.

\begin{theorem}\label{UpperBound}
There exists a quantum bit commitment protocol that uses a weak coin flipping protocol with cheating probability $1/2+\epsilon$ as a subroutine and achieves cheating probabilities less than $0.739+O(\epsilon)$.
\end{theorem}

Our protocol is a quantum improvement of the following simple protocol that achieves cheating probability $3/4$. Alice commits to bit $b$ by preparing the state $1/\sqrt{2}(\ket{bb}+\ket{22})$ and sending the second qutrit to Bob. In the reveal phase, she sends the first qutrit and Bob checks that the pure state is the correct one. It is not hard to prove that both Alice and Bob can cheat with probability $3/4$~\cite{Amb01,KN04}. The main idea in order to reduce the cheating probabilities for both players is the following: first we increase a little bit the amplitude of the state $\ket{22}$ in this superposition. This decreases the cheating probability of Bob. However, now Alice can cheat even more. To remedy this, we use the quantum procedure of a weak coin flipping so that Alice and Bob jointly create the above initial state (with the appropriate amplitudes) instead of having Alice create it herself. We present now the details of the protocol.

\subsection{The protocol}

\paragraph{Commit phase, Step 1}
Alice and Bob perform an unbalanced weak coin flipping procedure (without measuring the final outcome), where Alice wins with probability $1-p$ and Bob with probability $p$. As we said, we can think of this procedure as a big unitary operation that creates a joint pure state in the space of Alice and Bob. Moreover, Alice and Bob have each a special 1-qubit register that they can measure at the end of the protocol in order to read the outcome of the weak coin flipping. Here, we assume that they don't measure anything and that at the end  Alice sends back to Bob all her garbage qubits. In other words, in the honest case, Alice and Bob share the following state at the end of the weak coin protocol
\begin{align*}
\ket{\Omega} = \sqrt{p}  \ket{L}_{\spa{A}} \otimes \ket{L,G_L}_{\spa{B}} + \sqrt{1-p}  \ket{W}_{\spa{A}} \otimes \ket{W,G_W}_{\spa{B}}
\end{align*}
where $W$ corresponds to the outcome "Alice wins" and $L$ corresponds to the outcome "Alice loses". The spaces $\spa{A},\spa{B}$ correspond to Alice's and Bob's private quantum space.  The garbage states $\ket{G_W},\ket{G_L}$ are known to both players.

\paragraph{Commit phase, Step 2}

After the end of the weak coin flipping procedure, Alice does the following. Conditioned on her qubit being $W$, she creates two qutrits in the state $\ket{22}$ and sends the second to Bob. Conditioned on her qubit being $L$, she creates two qutrits in the state $\ket{bb}$ where $b$ is the bit she wants to commit to and sends the second to Bob. If the players are both honest, they share the following state:
\begin{align*}
\ket{\Omega_b} = \sqrt{p}  \ket{L,b}_{\spa{A}} \otimes \ket{L,b,G_L}_{\spa{B}} + \sqrt{1-p}  \ket{W,2}_{\spa{A}} \otimes \ket{W,2,G_W}_{\spa{B}}
\end{align*}

\COMMENT{
Let's $\ket{\overline{2}} = \ket{2,0,\phi_0}$, $\ket{\overline{b}} = \ket{b,1,\phi_1}$ for $b \in \zo$. The first register corresponds to Alice's last message of the commit phase, the second register corresponds to the outcome of the original coin flip and the third register corresponds to the garbage of this coin flip given to Bob. We have that
$\ket{\Omega_b} = \sqrt{p}\ket{0}_{\spa{A}}\ket{\overline{2}}_{\spa{B}} + \sqrt{1-p}\ket{1}_{\spa{A}}\ket{\overline{b}}_{\spa{B}}$. In the honest case, Bob hence has the state $\sigma_b = p \ketbra{\overline{2}}{\overline{2}} + (1-p) \ketbra{\overline{b}}{\overline{b}}$ depending on the bit $b$ Alice commits to.
}

\paragraph{Reveal phase}
In the reveal phase, Alice sends $b$ and all her remaining qubits in space $\spa{A}$ to Bob. Bob checks that he has the state $ \ket{\Omega_b}$.

\subsection{Analysis}
If Alice and Bob are both honest then Alice always successfully reveals the bit $b$ she committed to.
\paragraph{Cheating Bob}
Bob is not necessarily honest in the weak coin flipping protocol, however the weak coin flipping has small bias $\epsilon$. Since Alice is honest, Bob has all the qubits expect the one qubit which is in Alice's output register. At the end of the first step of the Commit phase, Alice and Bob share a state
\begin{align*}
\ket{\Omega^*} = \sqrt{p'} \ket{L}_{\spa{A}} \ket{\Psi_L}_{\spa{B}} + 
\sqrt{1 - p'} \ket{W}_{\spa{A}} \ket{\Psi_W}_{\spa{B}}
\end{align*}
for some states $\ket{\Psi_L},\ket{\Psi_W}$ held by Bob.
Recall that the outcome $L$ in Alice's output register corresponds to the outcome where Alice loses the weak coin flipping protocol. Hence, for any cheating Bob, since our coin flipping has bias $\eps$, we have $p' \le p + \eps$. At the end of the commit phase, depending on Alice's committed bit $b$, the joint state is 
\[
\ket{\Omega^*_b} = \sqrt{p'} \ket{L,b}_{\spa{A}} \ket{b,\Psi_L}_{\spa{B}} + 
\sqrt{1 - p'} \ket{W,2}_{\spa{A}} \ket{2,\Psi_W}_{\spa{B}}
\]
and Bob's density matrix is
\[
\sigma^*_b = p' \ketbra{b,\Psi_L}{b,\Psi_L} + (1-p')\ketbra{2,\Psi_W}{2,\Psi_W}. 
\]
By Proposition \ref{Hel67}, we have
\begin{align*}
P^*_B & = \Pr[\mbox{ Bob guesses } b] \leq \frac{1}{2} + \frac{\Delta(\sigma^*_0,\sigma^*_1)}{2} = \frac{1}{2} + \frac{p'}{2} \le \frac{1+p}{2}+ \frac{\eps}{2}
\end{align*}

\paragraph{Cheating Alice} 

Let $\sigma_b$ be Bob's reduced state at the end of the commit phase when both players are honest. Let $\ket{\overline{x}}=\ket{L,x,G_L}$ for $x \in \zo$ and $\ket{\overline{2}}=\ket{W,2,G_W}$. We have
\COMMENT{
\[ \sigma_b = p  \ketbra{L,b,G_L}{L,b,G_L} + (1-p)  \ketbra{W,2,G_W}{W,2,G_W}
\]
}
\[
\sigma_b = p  \ketbra{\overline{b}}{\overline{b}} + (1-p)  \ketbra{\overline{2}}{\overline{2}}
\]
Let $\xi$ be Bob's state at the end of the commit phase for a cheating Alice. Let $r_i = \triple{\overline{i}}{\xi}{\overline{i}}$ for $i \in \{0,1,2\}$. From the characterization of the fidelity in Proposition~\ref{CPTPfidelity}, we have that 
\begin{align*}
F(\xi,\sigma_b) \le \sqrt{pr_b} +\sqrt{(1-p)r_2}
\end{align*}
From standard analysis of bit commitment protocol (for example~\cite{KN04} ), we have using Uhlmann's Theorem that 
\begin{align*}
P^*_A & \le \frac{1}{2}\left(F^2(\xi,\sigma_0) +F^2(\xi,\sigma_1)\right) \\
& \le \frac{1}{2}\left(\sqrt{pr_0} +\sqrt{(1-p)r_2}\right)^2 + \frac{1}{2}\left(\sqrt{pr_1} +\sqrt{(1-p)r_2}\right)^2 
\end{align*}

In order to get a tight bound for the above expression, we use here the property of the weak coin flipping. Recall that $\ket{\overline{2}} = \ket{W,2,G_W}$ has its first register as $W$ (this corresponds to Alice winning the coin flip). On the other hand, $\ket{\overline{0}}$ and $\ket{\overline{1}}$ have $L$ as their first register, corresponding to the case where Bob wins.
For any cheating Alice, she can win the weak coin flip with probability smaller than $1-p +\eps$ and hence this means in particular that $r_2 \le 1-p + \eps$. Moreover, $r_0 +r_1 + r_2 \le 1$.
For $\eps < p(1 - \frac{1}{2 - p})$ , we can show that this quantity is maximal when $r_2$ is maximal and $r_0 = r_1=(p-\eps)/2$ (proven in Appendix~\ref{AppLowerBound}). This gives us
\begin{align*}
P^*_A &\le \left( \sqrt{p \cdot \frac{p - \eps}{2}} + \sqrt{(1-p)(1-p+\eps)} \right)^2
 \le \left(1 - (1 - \frac{1}{\sqrt{2}})p\right)^2 + O(\eps)
\end{align*}

\paragraph{Putting it all together}
Except for the terms in $\eps$, we obtain exactly the same quantities as in our lower bound. By equalizing these cheating probabilities, we have
\[
\max\{P^*_A,P^*_B\} \approx 0.739 + O(\eps)
\]
Since we can have $\eps$ arbitrarily close to $0$ (Proposition \ref{Mochon}) and we can have an unbalanced weak coin flipping protocol with probability arbitrarily close to $p$ (Proposition \ref{CK09}), we conclude that our protocol is arbitrarily close to optimal.


\section{Proof of the classical lower bound}
In this Section, we show a $3/4$ lower bound for classical bit commitment schemes when players additionally have the power to perform perfect (strong or weak) coin-flipping. This will show that unlike strong coin flipping, quantum and classical bit commitment are not alike in the presence of weak coin flipping. 

 We first describe such protocols in Section~$\ref{Description}$. In Section~$\ref{Cheating}$, we construct a cheating strategy for Alice and Bob for these protocols such that one of the players can cheat with probability at least $3/4$.
\subsection{Description of a classical bit commitment protocol with perfect coin flips}\label{Description}
We describe classical bit commitment schemes when players additionally have the power to perform perfect (strong or weak) coin-flipping. The way we deal with the coin is the following: when Alice and Bob are honest, they always output the same random value $c$ and both players know this value. We can suppose equivalently that a random coin $c$ is given publicly to both Alice and Bob each time they perform coin flipping. We describe any BC protocol with coins as follows:
\begin{itemize}
\item Alice and Bob have some private randomness $R_A$ and $R_B$ respectively.
\item \emph{Commit phase:} Alice wants to commit to some value $x$. Let $N$ the number of rounds of the commit phase. For $i = 1$ to $N$: Alice sends a message $a_i$, Bob sends a message $b_i$, Alice and Bob flip a coin and get a public $c_i \in_R \zo$.
\item \emph{Reveal phase:} Alice wants to decommit to some value $y$ ($ = x$ if Alice is honest). 
\begin{enumerate}
\item Alice first reveals $y$. This is a restriction for the protocol but showing a lower bound for such protocols will show a lower bound for all protocols since this can only limit Alice's cheating possibilities without helping Bob.
\item Let $M$ the number of rounds of the reveal phase. For $i = 1$ to $M$: Alice sends a message $a'_i$, Bob sends a message $b'_i$, Alice and Bob flip a coin and get a public $c'_i \in_R \zo$.
\item Bob has an accepting procedure $Acc$ to decide whether he accepts the revealed bit or whether he aborts (if Bob catches Alice cheating).
\end{enumerate}
\end{itemize}

We denote the commit phase transcript by 
$t_C = (a_1,b_1,c_1,\dots,a_N,b_N,c_N)$. If Alice and Bob are honest, then we can write $t_C = T_C(R_A,R_B,c,x)$ where $T_C$ is a function fixed by the protocol that takes as input Alice and Bob's private coins $R_A,R_B$, the outcomes of the public coin flips $c = (c_1,\dots,c_N)$ as well as the bit $x$ Alice wants to commit to and outputs a commit phase transcript $t_C$.  If we can write $t_C = T_C(R_A,R_B,c,x)$ for some $R_A,R_B,c,x$, we say that $t_C$ is an honest commit phase transcript.

Similarly, we define the  decommit phase transcript by $t_D = (a'_1,b'_1,c'_1,\dots,a'_M,b'_M,c'_M)$. If Alice and Bob are honest, we can write $t_D = T_D(R_A,R_B,c',y,t_C)$, where $T_D$ is a function fixed by the protocol that takes as input Alice and Bob's private coins $R_A,R_B$, the outcomes of the public coin flips $c' = (c'_1,\dots,c'_M)$, the bit $y$ Alice reveals as well as the commit phase transcript $t_C$ and outputs a reveal phase transcript $t_D$.
 If we can write $t_D = T_D(R_A,R_B,c',y,t_C)$ for some $R_A,R_B,c',y$ and some honest commit phase transcript $t_C$, we say that $t_D$ is an honest reveal phase transcript.

Whether Bob accepts at the end of the protocol depends on both transcripts $t_C,t_D$ of the commit and reveal phase, the bit $y$ Alice reveals as well as Bob's private coins. We write that $Acc(t_C,t_D,y,R_B) = 1$ when Bob accepts.

Note that in the honest case, Bob always accepts Alice's deommitment. This means that we can transform Alice's honest strategy in the reveal phase to a deterministic strategy which will also be always accepted. This fact will be useful in the proof.

\subsection{Proof of the classical lower bound}\label{Cheating}
In this Section, we construct cheating strategies for Alice and Bob such that one of the players will be able to cheat with probability greater than $3/4$. We only consider cheating strategies where Alice and Bob are honest during the coin flips so again, they will be modeled as public and perfectly random coins. Moreover, Alice and Bob will always be honest during the commit phase. 

Before describing the cheating strategies we need some definitions. More particularly, we consider a cheating Alice who cheats during the reveal phase by following a deterministic strategy $A^*$. For a fixed honest commit phase transcript $t_C$, we can write the transcript of the reveal phase as a function of $A^*,R_B,c',y,t_C$, more precisely $T_D^*(A^*,R_B,c',y,t_C)$.

\begin{definition}
We say that $R_B$ is consistent with $t_C$ if and only if there exist $R_A,c,x$ such that $t_C = T_C(R_A,R_B,c,x)$.
\end{definition}

\begin{definition}
Let $t_C$ an honest commit phase transcript. We say that $t_C \in A_y$ if and only if
\begin{align*}
\exists A^* \textrm{ s.t. } \forall c' \textrm{ and } \forall R_B \textrm{ consistent with } t_C,
 Acc(t_C,T_D^*(A^*,R_B,c',y,t_C),y,R_B) = 1
\end{align*}
\end{definition}
Intuitively, $t_C \in A_y$ means that if Alice and Bob output an honest commit phase transcript $t_C$, there is a deterministic strategy $A^*$ for Alice that allows her to reveal $y$ without Bob aborting, independently of Bob's private coins $R_B$. Since there is always a deterministic honest strategy for Alice in the reveal phase (when Alice and bob have been honest in the commit phase), we have 
\begin{align*}
\forall \ R_A,R_B,c,x \;\;\; \ \ t_C = T_C(R_A,R_B,c,x) \in A_x
\end{align*}
Notice also that for any honest commit phase transcript $t_C$, both players Alice and Bob can compute whether $t_C \in A_u$ for both $u=0$ and $u=1$.

\begin{definition} We define the probability
\[
p_u = \Pr [t_C = T_C(R_A,R_B,c,{u}) \in A_{\overline{u}}] \quad \mbox{where the probability is taken over uniform } R_A,R_B,c.
\]
Consider that Bob is honest. $p_u$ is the probability that if Alice behaves honestly in the commit phase and commits to $u$, she has a deterministic cheating strategy to reveal $\overline{u}$ which always succeeds (independently of $c',R_B$).
\end{definition}

\noindent
We can now describe and analyze our cheating strategies for Alice and Bob and prove our theorem
\begin{theorem}
For any classical bit commitment protocol with access to public perfect coins, one of the players can cheat with probability at least $3/4$.
\end{theorem}
\begin{proof}
Let us fix a bit commitment protocol. We describe cheating strategies for Alice and Bob.
\paragraph{Cheating Alice}
\begin{itemize}
\item Commit phase:
Alice picks $x \in_R \zo$ and she honestly commits to $x$ during the commit phase. 
\item Reveal phase: if Alice wants to reveal $x$, she just remains honest during the reveal phase. By completeness of the protocol, this strategy succeeds with probability $1$. If Alice wants to reveal $\overline{x}$, we know by definition of $p_x$ that she succeeds with probability at least $p_x$.  This gives us:
\begin{align*}
P^*_A \ge \frac{1}{2} + \frac{p_x}{2} \end{align*}
since Alice chooses $x$ at random, we have:
\begin{align*}
P^*_A \ge \frac{1}{2} + \frac{p_0 + p_1}{4} \end{align*}
\end{itemize}

\paragraph{Cheating Bob}
As Alice, Bob is honest in the commit phase. Let $x$ the bit Alice committed to. Since Alice and Bob are honest the commit-phase transcript is $t_C = T_C(R_A,R_B,c,x)$ for uniformly random $R_A,R_B,c$. As said before, we know that $t_C \in A_x$.

At the end of the commit phase, Bob wants to guess the bit $x$ Alice commits to and he performs the following strategy: if $t_C \in A_0 \cap A_1$  he guesses $x$ at random. If $\exists! \; u \textrm{ s.t. }  t_C \notin A_u$  he guesses $x = \overline{u}$. 

We know that Bob succeeds in cheating with probability $1/2$ if $t_C \in A_{\overline{x}}$ and with probability $1$ if $t_C \notin A_{\overline{x}}$. This gives us $P^*_B \geq  p_x \cdot \frac{1}{2} + (1-p_x)\cdot 1 = 1 - \frac{p_x}{2}$. Since again, Alice's bit $x$ is uniformly  random, we have  
\begin{align*}
P^*_B \geq 1 - \frac{p_0 + p_1}{4} \end{align*}

\paragraph{Putting it all together} Taking Alice and Bob cheating probabilities together, we have \\
$P^*_A + P^*_B \ge 3/2$ {which gives} $
\max\{P^*_A,P^*_B\} \ge 3/4$.

\end{proof}

\bibliography{paper}
\bibliographystyle{alpha}




\newpage

\begin{appendix}
\section{Proof of $r_0 = r_1$ and $r_2$ maximal in the quantum lower bound}~\label{AppLowerBound}
In this Section, we show the following:
\begin{proposition}
Let
\[ P^*_A \le \frac{1}{2}\left(\sqrt{pr_0} +\sqrt{(1-p)r_2}\right)^2 + \frac{1}{2}\left(\sqrt{pr_1} +\sqrt{(1-p)r_2}\right)^2 
\]
with the constraints: $r_0,r_1,r_2 \ge 0$, $r_0 + r_1 + r_2 \le 1$ and $r_2 \le 1 - p + \eps$ for $\eps < p(1 - \frac{1}{2-p})$. This cheating probability is maximized for $r_0 = r_1 = \frac{p-\eps}{2}$ and $r_2 = 1 - p + \eps$.
\end{proposition}
\begin{proof}
First note that the maximal cheating probability is achieved for $r_0 + r_1 + r_2 = 1$ since this cheating probability is increasing in $r_0,r_1,r_2$. 

We first show that $r_0 = r_1$. Let's fix $r_2$. This means that $S = r_0 + r_1 = 1 - r_2$ is fixed. Let $u = \sqrt{(1-p)r_2}$. We have
\begin{align*}
P^*_A \le f(r_0) = \frac{1}{2}\left(\sqrt{pr_0} + u\right)^2 + \frac{1}{2}\left(\sqrt{p(S-r_0)} + u\right)^2 .
\end{align*}
Taking the derivative, we have
\begin{align*}
f'(r_0) & = \frac{1}{2}\left(2\sqrt{p}\frac{1}{2\sqrt{r_0}} (\sqrt{pr_0} + u) - 2\sqrt{p}\frac{1}{2\sqrt{(S-r_0)}} (\sqrt{p(S-r_0)} + u)\right) \\
& = \frac{1}{2}\left(p + \frac{u\sqrt{p}}{\sqrt{r_0}} - p - \frac{u}{\sqrt{p}}{\sqrt{S-r_0}}\right) \\
& = \frac{u\sqrt{p}}{2}\left(\frac{1}{\sqrt{r_0}} - \frac{1}{\sqrt{S - r_0}}\right)
\end{align*}
We have $f'(r_0) > 0$ for $r_0 < S/2$ ; $f'(r_0) = 0$ for $r_0 = S/2$ ; $f'(r_0) < 0$ for $r_0 > S/2$. This means that the maximum of $f$ is achieved for $r_0 = S/2$ $\ie$ $r_0 = r_1$. \\

We now show that $r_2 = 1 - p + \eps$ gives the maximal cheating probability if $\eps$ is not too big. Since $P^*_A$ is maximal for $r_0 = r_1$ and for $r_0 + r_1 + r_2 = 1$, we have
\begin{align*}
P^*_A & \le \frac{1}{2}\left(\sqrt{pr_0} + \sqrt{(1-p)r_2}\right)^2 + \frac{1}{2}\left(\sqrt{pr_0} + \sqrt{(1-p)r_2}\right)^2 \\
& \le (\sqrt{pr_0} + \sqrt{(1-p)r_2})^2 \\
& \le \left(\sqrt{p(\frac{1 - r_2}{2})} + \sqrt{(1-p)r_2}\right)^2 = g(r_2)
\end{align*}
Again, we take the derivative of $g$.
\begin{align*}
g'(r_2) & = \left(- \frac{\sqrt{p}}{\sqrt{2(1-r_2)}} + \frac{\sqrt{1-p}}{{\sqrt{r_2}}}\right)\cdot\left(\sqrt{p(\frac{1 - r_2}{2})} + \sqrt{(1-p)r_2}\right)
\end{align*}
From this, we have
\begin{align*}
g'(r_2) \ge 0 & \Leftrightarrow \left(- \frac{\sqrt{p}}{\sqrt{2(1-r_2)}} + \frac{\sqrt{1-p}}{{\sqrt{r_2}}}\right) \ge 0 \\
& \Leftrightarrow \sqrt{\frac{p}{2(1-r_2)}} \le \sqrt{\frac{1 - p}{r_2}} \\
& \Leftrightarrow pr_2 \le 2(1-r_2)(1-p) \\
& \Leftrightarrow r_2 \le 1 - \frac{p}{2-p} 
\end{align*}
For $\eps < p(1 - \frac{1}{2-p})$, we have $1 - p + \eps < 1 - \frac{p}{2-p}$, so when $\eps < p(1 - \frac{1}{2-p})$, $g(r_2)$ is always increasing when $r_2 \le 1 - p + \eps$ and is maximal when $r_2 = 1 - p + \eps$, which concludes the proof.
\end{proof}
\end{appendix}

\end{document}